\documentclass[a4paper,10pt]{article}

\usepackage[utf8]{inputenc}
\usepackage[T1]{fontenc}

\usepackage{a4wide}

\usepackage{amsmath,amssymb}
\usepackage{amsthm}
\usepackage{enumerate} 

\usepackage{authblk}

\usepackage[dvipsnames,table,xcdraw,svgnames]{xcolor}
\definecolor{highlightNEW}{named}{black}
\definecolor{mylilas}{RGB}{170,55,241}
\definecolor{mygreen}{RGB}{70,180,5}
\definecolor{myred}{RGB}{244,63,43}

\usepackage{hyperref}
\hypersetup{
    breaklinks=true,
    colorlinks=true,
    linkcolor=Navy,
    citecolor=Navy,
    urlcolor=Navy
}

\usepackage{graphicx}
\graphicspath{{fig/}}
\usepackage{epstopdf}
\usepackage[section]{placeins} 

\usepackage{booktabs}
\usepackage{multirow}

\usepackage[absolute,overlay,showboxes]{textpos}
\TPMargin{2mm}
\textblockrulecolour{Navy}

\usepackage{natbib}

\newcommand{\Citeauthornohref}[1]{{\protect\NoHyper\Citeauthor{#1}\protect\endNoHyper}}

\usepackage{listings}
\usepackage{soulutf8}

\soulregister\cite7%
\soulregister\citep7%
\soulregister\eqref7%
\soulregister\pageref7%
\soulregister\ref7%

\newtheorem{theorem}{Theorem}[section] 
 
\newtheorem{definition}[theorem]{Definition} 
 
\newtheorem{lemma}[theorem]{Lemma}

\newcommand{\doi}[1]{DOI~\href{\detokenize{http://dx.doi.org/#1}}{\detokenize{#1}}}
\newcommand{\zblnumber}[1]{Zbl~\href{\detokenize{https://zbmath.org/?q=an:#1}}{\detokenize{#1}}}
\newcommand{\mrnumber}[1]{\href{\detokenize{https://www.ams.org/mathscinet-getitem?mr=#1}}{\detokenize{MR#1}}}

\newcommand{\AAE}{\operatorname{AAE}}
\renewcommand{\AE}{\operatorname{AE}}
\newcommand{\ARE}{\operatorname{ARE}}
\newcommand{\RE}{\operatorname{RE}}
\newcommand{\B}{\mathcal{B}}

\renewcommand{\d}{\,\mathrm{d}}
\newcommand{\dx}{\,\mathrm{d}x}
\newcommand{\dv}{\,\mathrm{d}v}
\newcommand{\e}{\mathrm{e}}
\newcommand{\E}{\mathbb{E}}
\newcommand{\F}{\mathcal{F}}

\newcommand{\N}{\mathbb{N}}
\renewcommand{\P}{\mathbb{P}}

\newcommand{\R}{\mathbb{R}}

\newdimen\CdotAxis
\newcommand*{\CdotAux}[3]{%
  {%
    \settoheight\CdotAxis{$#2\vcenter{}$}%
    \sbox0{%
      \raisebox\CdotAxis{%
        \scalebox{#1}{%
          \raisebox{-\CdotAxis}{%
            $\mathsurround=0pt #2#3$%
          }%
        }%
      }%
    }%
    \dp0=0pt %
    \sbox2{$#2\bullet$}%
    \ifdim\ht2<\ht0 %
      \ht0=\ht2 %
    \fi
    \sbox2{$\mathsurround=0pt #2#3$}%
    \hbox to \wd2{\hss\usebox{0}\hss}%
  }%
}
\makeatletter
\def\mathcolor#1#{\@mathcolor{#1}}
\def\@mathcolor#1#2#3{%
  \protect\leavevmode
  \begingroup
    \color#1{#2}#3%
  \endgroup
}
\makeatother

\let\oldalpha\alpha
\renewcommand{\alpha}{\mathcolor{highlightNEW}{\oldalpha}}
\newcommand{\ccode}[2]{\par
        \vspace*{8pt}
        {{\leftskip18pt\rightskip\leftskip
        \noindent{\it #1}\/: #2\par}}\par}
\newcommand{\keywords}[1]{\ccode{Keywords}{#1}}
\newcommand{\email}[1]{\href{mailto:#1}{#1}}

\def\received#1{Received~#1\par}
\def\revised#1{Revised~#1\par}

\usepackage[mathscr]{euscript}
\DeclareSymbolFont{rsfs}{U}{rsfs}{m}{n}
\DeclareSymbolFontAlphabet{\mathscrsfs}{rsfs}

\newcommand{\jpTitle}{Solution of option pricing equations using orthogonal polynomial expansion}
\newcommand{\jpAuthors}{F. Baustian, K. Filipov\'{a} and J. Posp\'{\i}\v{s}il}
\newcommand{\jpKeywords}{orthogonal polynomial expansion; Hermite polynomials; Laguerre polynomials; Heston model; option pricing}
\newcommand{\jpMSC}{33C45; 65M60; 91G20; 91G60}%
\newcommand{\jpJEL}{C58; G12; C63}%
\newcommand{\jpDateReceived}{13 December 2019} 
\newcommand{\jpDateRevised}{27 February 2020, 23 June 2020}
\newcommand{\jpDate}{}

\author[1]{Falko Baustian} 
\author[2]{Kate\v{r}ina Filipov\'{a}} 
\author[2]{Jan Posp\'{\i}\v{s}il\thanks{Corresponding author, \email{honik@kma.zcu.cz}}} 
\affil[1]{Department of Mathematics, University of Rostock, Ulmenstra\ss e 69, 18057 Rostock, Germany}
\affil[2]{NTIS - New Technologies for the Information Society, Faculty of Applied Sciences, \authorcr University of West Bohemia, Univerzitn\'{\i} 2732/8, 301 00 Plze\v{n}, Czech Republic,\vspace*{3pt}}

\ifpdf
\hypersetup{
  pdftitle={\jpTitle},
  pdfauthor={\jpAuthors},
  pdfkeywords={\jpKeywords},
  pdfinfo={
      MSC={\jpMSC},
      JEL={\jpJEL}
  }
}
\fi

\title{\textcolor{Navy}{\textsc{\jpTitle}}}
\date{\jpDate}

\begin{document}

\maketitle

\begin{center}
\received{\jpDateReceived}
\revised{\jpDateRevised}
\end{center}

\begin{abstract}
In this paper we study both analytic and numerical solutions of option pricing equations using systems of orthogonal polynomials.
Using a Galerkin-based method, we solve the parabolic partial differential equation for the Black-Scholes model using Hermite polynomials and for the Heston model using Hermite and Laguerre polynomials. 
We compare obtained solutions to existing semi-closed pricing formulas. 
Special attention is paid to the solution of Heston model at the boundary with vanishing volatility.

\end{abstract}

\ccode{MSC classification}{\jpMSC}
\ccode{JEL classification}{\jpJEL}

\setcounter{tocdepth}{2}
\tableofcontents
\clearpage

\section{Introduction}\label{sec:introduction}

One of the fundamental tasks in financial mathematics is the pricing of derivatives, in particular option pricing. An \emph{option} is a contract between two parties which gives the holder the right (but not the obligation) to buy or sell the underlying asset under certain conditions on or before a specified future date. The price that is paid for the underlying when the option is exercised is called \emph{strike} price and the last day on which the option may be exercised is called expiration date or \emph{maturity} date. Whether the holder has the right to buy or sell the underlying asset depends on the type of option to which the contract is signed. There is either a \emph{call option} which allows the holder to buy the asset at a stated price within a specific time-frame or a \emph{put option} which allows the holder to sell the asset. In this article we will restrict ourselves to European options that can be exercised only on the expiration date.

In their Nobel-prize winning paper, \cite{BlackScholes73} proposed a model for evaluating the fair value of the European call option that gives the right to buy a single share of common stock and derived a semi-closed formula for the option price, the so-called Black-Scholes formula. For the model they have assumed a frictionless market with ideal conditions like the absence of arbitrage and the possibility to borrow and lend any amount of money and to buy and sell any amount of stock, respectively. Volatility in the Black-Scholes (BS) model is assumed to be constant which has later become its most discussed feature. Constant volatility matches poorly with the observed implied volatility surface for real market data. Especially for out of the money options the market prices are significantly higher than what the model suggests. This phenomenon is widely known as the volatility smile. For a better fit to the data, \cite{HullWhite87} proposed to model volatility as another stochastic process. There are various stochastic volatility models from \cite{HullWhite87}, \cite{SteinStein91}, \cite{Heston93}, and many others. Later on, additional jump components were included into the models, e.g. \cite{Bates96}. 

Up to this day, the Heston model is quite popular among economists and practitioners. \cite{Heston93} modelled the volatility using the mean-reverting \cite{CIR85} process (CIR), which allowed arbitrary correlation between volatility and spot asset returns. Heston also derived a semi-closed formula close to the BS formula. Both in BS and Heston model, one can derive the pricing partial differential equation (PDE) in several different ways, for example \citep{Wilmott98,Rouah13,Hull18} using arbitrage arguments with self-financing trading strategies, approaches with martingale measures or the Fokker-Planck equation for the transition probability density function. Although semi-closed formulas have been widely used in practice for a long time, only recently \cite{DanekPospisil20ijcm} showed that for certain values of model parameters these formulas can bring serious numerical difficulties especially in evaluation of the integrands in these formulas and their implementation therefore sometimes requires a demanding high precision arithmetic to be adopted.

Many different numerical methods can be used to solve option pricing problems such as Monte Carlo methods (including the Quasi Monte Carlo), Fourier based methods (including the Fast Fourier Transform method, Fourier method with Gauss-Laguerre quadrature, cosine series method), finite differences methods (with different time-stepping schemes, different grid refinements including the adaptive refinement or discontinuous Galerkin method), finite element methods (including the method with NURBS basis functions introduced by \cite{PospisilSvigler19ijcm}) or for example radial basis function methods (RBF). We refer the reader to the references in the BENCHOP project report written by \cite{Sydow15}, who implemented fifteen different numerical methods with the help of different advanced specialized techniques, and who compared all methods for different benchmark problems and consequently discussed advantages and disadvantages of each method.

The aim of this paper is to solve the pricing PDEs for both BS and Heston model using orthogonal polynomial expansions that are motivated by the Galerkin method. The expansion approach offers several advantages as we approximate the solution by smooth functions. Therefore, it gives more insight into how parameters influence prices and to what extent and hence give a better understanding of the solution than the semi-closed form or other approximation method especially for the Heston model. For the sake of clarity of the method we omit application of specialized techniques that could further improve the proposed method. Among the other mentioned methods, only FEM with smooth basis functions and RBF approximate the solution by smooth functions. One advantage of the orthogonal polynomial expansion is hence the independence of the space variable discretization (finite elements) or spacial node locations (RBF).

\cite{Aubin67} studied Galerkin type methods and their convergence for elliptic partial differential equations and \cite{Birkhoff68} used piecewise Hermite polynomials for this problem. Time-dependent equations were investigated with the usage of the Galerkin method by \cite{Swartz69}. The initial value problem for a general parabolic equation of second order was first studied by \cite{Douglas70}. They used Galerkin type methods, both continuous and discrete in time, and established a priori estimates to control the error. These articles initiated several other papers by \cite{Dupont72}, \cite{Fix72}, \cite{Wheeler73}, \cite{Bramble74}, \cite{Bramble77}, and \cite{Thomee77}. Most of the a priori estimates are formulated with regard to the $L^2$ norm but \cite{Bramble77} offers estimates for the maximum norm, as well. Nonlinear parabolic equations were covered by \cite{Wheeler73}. A survey of results can be found in \cite{Thomee78} and in the monograph \cite{Thomee06}.

The application of orthogonal polynomial expansions in option pricing was to our knowledge for the first time suggested by \cite{Jarrow82} who pioneered the use of Edgeworth expansions for valuation of derivative securities. Later \cite{Corrado96} introduced the Gram-Charlier expansions. In the recent past, Hermite polynomial expansion approaches have been used in some interesting articles regarding different aspects of the option pricing problem.

\cite{Xiu14} studied a closed-form series expansion of European call option prices in the time variable and this series expansion was derived using the Hermite polynomials. Xiu introduced two approaches on vanilla option and binary option. The first one has been a bottom-up Hermite polynomial approach and the second one has been a top-down lucky guess approach. As the benchmark model he has chosen BS model but stated that square-root (SQR) models for the volatility like \cite{Heston93}, quadratic volatility (QV) models, constant elasticity of variance (CEV) models, which introduces one additional parameter the elasticity of variance, or several jump-diffusion models can be considered, see for example a recent monograph by \cite{Lewis16}.

\cite{Heston17} showed that Edgeworth expansions for option valuation are equivalent to approximating the option payoff using Hermite polynomials and logistic polynomials. Consequently, the value of an option is equal to the value of an infinite series of replicating polynomials. Heston and Rossi provide efficient alternative moment-based formulas to express option values in terms of skewness, kurtosis and higher moments.

Polynomial expansions with Hermite and Laguerre polynomials play also a substantial role in \cite{Alziary18}. The authors rigorously formulate the Cauchy problem connected to the Heston model as a parabolic PDE with a special focus on the boundary conditions which are often neglected in the literature. \Citeauthornohref{Alziary18} provide the real analyticity of the solution which is directly connected to the problem of market completeness studied in \cite{Davis08}. The polynomial expansions are used in the proof of the main results of the article. Further investigations of the boundary conditions can be found in the forthcoming work \cite{Alziary20pre}.

Very recently, option pricing with orthogonal polynomial expansions has been studied by \cite{Ackerer20}, who derived option prices series representation by expansion of the characteristic functions rather than by solving the pricing PDE.

The structure of the paper is the following. In Section \ref{sec:preliminaries} we introduce system of orthogonal polynomials, studied models as well as other necessary terms and fundamental properties. In Section \ref{sec:methodology} we solve the Black-Scholes and Heston PDE using the orthogonal polynomial expansion. To solve the BS PDE we use Hermite polynomials and to solve the Heston PDE we use a combination of Hermite and Laguerre polynomials. In Section \ref{sec:results} we present all numerical results, especially comparison to the existing semi-closed form solutions. We conclude in Section \ref{sec:conclusion}.

\section{Preliminaries and notation}\label{sec:preliminaries}

\subsection{Orthogonal polynomials}

Standard theory for parabolic PDEs requires initial data in a Lebesgue space. In the PDE pricing approach for European-type derivatives the initial value corresponds to the payoff function of the contract but unfortunately the payoff of many European options, e.g., the European call option, is unbounded and not Lebesgue-integrable. For this reason we consider weighted Lesbesgue spaces with a positive \emph{weight function} $w$ as studied in \cite{Kufner80}, \cite{Kufner87}, \cite{Funaro92}.

The \emph{weighted Lebesgue space} $L^2(\R, w\d x)$ is the space of all measurable functions $f$ for which
\begin{align*}
\Vert f \Vert_{w} := \left( \int_{\R} \vert f(x) \vert^2 w(x) \d x \right)^{1/2} < \infty.
\end{align*}
As usual, we consider representatives of classes of functions which are equal almost everywhere. We can also define weighted Sobolev spaces $H^k(\R, w\d x)$ for $k\geq1$. Again, we refer the reader to \cite{Kufner80}, \cite{Kufner87}, and \cite{Funaro92}, for details about such spaces.

We consider sequences $(F_n)$ of real polynomials in $L^2(\R,w\dx)$ which are pairwise orthogonal with respect to the \emph{inner product} defined by
\begin{equation}\label{def:inner}
\langle f, g\rangle_{w} := \int_{\R} f (x) g(x) w(x) \dx \quad\mbox{for}\ f,g\in L^2 (\R, w\d x).
\end{equation}
It can be shown that for given $F_0(x)$ and $F_1(x)$ that are not both identically zero, there exist functions $\alpha(n,x)$ and $\beta(n,x)$ such that the system of orthogonal polynomials satisfies the so called \emph{three-term recurrence relation}
\begin{equation}\label{e:recurrence_general}
F_{n+1}(x)=\alpha (n,x) F_n(x) + \beta (n,x) F_{n-1}(x),\quad n\in\N.
\end{equation}
The relation \eqref{e:recurrence_general} is arguably the single most important piece of information for the constructive and computational use of orthogonal polynomials. For more details about general systems of orthogonal polynomials and on the proof of the recurrence relation we refer the reader to the book by \cite{Gautschi04}.

Throughout the paper we will work especially with Hermite and Laguerre polynomials. Their properties are rather extensively mentioned in many monographs, we refer the readers for example to the books by 
\cite[chap. 22]{AbramowitzStegun64},  
\cite[chap. 4]{Lebedev65}, 
\cite[chap. 5]{Szego75},
\cite{Thangavelu93} 
and \cite[chap. 18]{Olver10} to name a few. The definition and basic properties of Hermite and Laguerre polynomials can be found in all of these monographs.

\subsubsection{Hermite polynomials}

Hermite polynomials are orthogonal polynomials on the real line. There exists two types of Hermite polynomials that differ slightly in the choice of weight function and that are called \emph{probabilists'} (weight function $\e^{-{x^2}/2}$) and \emph{physicists'} (weight function $\e^{-x^2}$) \emph{Hermite polynomials}. Those two types can be easily converted into each other and we will consider physicists' polynomials only.

\begin{definition}\label{def:Her} 
The \emph{system of Hermite polynomials} is defined by the Rodrigues formula
\begin{align*}
H_m(x) &:= (-1)^m \e^{x^2} \frac{\d^m}{\d x^m} \e^{-x^2},\quad m\in\N_0.
\end{align*}
\end{definition}

The three-term recurrence \eqref{e:recurrence_general} for Hermite polynomials reads
\begin{equation}
\label{Her:recur}
H_{m+1}(x) = 2x H_m(x) - 2m H_{m-1}(x), \quad m \geq 1.
\end{equation}

The Hermite polynomials form a complete orthogonal system in the weighted Lebesgue space $L^2(\R,\e^{-x^2}\,\d x)$ with $\left\langle H_m,H_n\right\rangle_w=2^nn!\sqrt{\pi}\cdot\delta_{m,n}$, where $\delta_{m,n}$ is the Kronecker delta, as well as an orthogonal set in the weighted Sobolev space $H^k(\mathbb{R},\mathrm{e}^{-x^2}\,\mathrm{d}x)$ for $k\geq1$. See \cite[Sec. 4.14]{Lebedev65} for the orthogonality and \cite[Sec. 5.7]{Szego75} for the completeness of the system, respectively.

In the following lemma we state several useful simplifications of integral terms that are consequences of Definition \ref{def:Her} and the three-term recurrence \eqref{Her:recur}. 

\begin{lemma}\label{lem25}
For all $m, n \in \N_0$,
\begin{align}
\frac{1}{2^n n! \sqrt{\pi}} \int_{-\infty}^{\infty} H_m' (x) H_n' (x)\e^{-x^2} \d x &= 2m \delta_{m,n}, \label{IH1} \\
\frac{1}{2^n n! \sqrt{\pi}} \int_{-\infty}^{\infty} H_m (x) H_n' (x)\e^{-x^2} \d x &= \delta_{m+1,n}, \label{IH2} \\
\frac{1}{2^n n! \sqrt{\pi}} \int_{-\infty}^{\infty} x H_m' (x) H_n (x)\e^{-x^2} \d x &= 2(n+1)m \delta_{m,n+2} + m \delta_{m,n}, \label{IH4} \\
\frac{1}{2^n n! \sqrt{\pi}} \int_{-\infty}^{\infty} x H_m (x) H_n (x)\e^{-x^2} \d x &= \frac12 \delta_{m+1,n} + m \delta_{m-1,n}. \label{IH5}
\end{align}
\end{lemma}
A proof can be found in the thesis \cite[chap. 2, Lemmas 2.5--2.9]{Filipova19}.

\subsubsection{Laguerre polynomials}

The volatility process in the Heston model is strictly positive provided that the Feller condition is satisfied. Hence, we need a system of orthogonal polynomials on the positive part of the real line for the expansion in the volatility variable. With the weight function $w\colon\R^+\to\R$, $w(v)=\e^{-v}$, on $\R^+=(0,\infty)$ such a system is given by the Laguerre polynomials.

\begin{definition}\label{def:Lag}
The \emph{system of Laguerre polynomials} is defined by
\begin{align*}
L_n(v) &:= \frac{\e^v}{n!} \frac{\d^n}{\d v^n}\left(\e^{-v}v^n\right), \quad  n\in\N_0.
\end{align*}
\end{definition}

The three-term recurrence \eqref{e:recurrence_general} for the Laguerre polynomials is
\begin{equation}
\label{Lag:recur}
L_{n+1}(v) = \frac1{n+1}\left[ (-v+2n+1)L_n(v) - nL_{n-1} (v) \right],\quad n \geq 1.
\end{equation}

The Laguerre polynomials form a complete orthonormal system in the weighted Lebesgue space $L^2(\R^{+},\e^{-v}\,\d v)$. The orthonormality of the system is studied in \cite[Sec. 4.21]{Lebedev65} and the completeness in \cite[Sec. 5.7]{Szego75}.

We use Definition \ref{def:Lag} and the three-term recurrence \eqref{Lag:recur} to obtain some simplifications.

\begin{lemma}\label{lem28}
For all $m, n \in \N_0$,
\begin{align}
\int_{0}^{\infty} v L_m (v) L_n (v) \e^{-v} \d v &= (2m+1) \delta_{m,n} - m \delta_{m-1,n} - (m+1) \delta_{m+1,n}, \label{IL1} \\
\int_{0}^{\infty} v L_m'(v) L_n(v) \e^{-v} \d v &= m (\delta_{m,n} - \delta_{m-1,n}), \label{IL2} \\
\int_{0}^{\infty} v L_m'(v) L_n'(v) \e^{-v} \d v &= m \delta_{m,n}, \label{IL3} \\
\int_{0}^{\infty} L_m'(v) L_n(v) \e^{-v} \d v &= -\sum_{a=0}^{m-1} \delta_{a,n}. \label{IL4} 
\end{align}
\end{lemma}
A proof can be found in the thesis \cite[chap. 2, Lemmas 2.12--2.15]{Filipova19}. 
It is worth mentioning that the formulas in Lemma \ref{lem25} and Lemma \ref{lem28} are not stated in any of the monographs listed above.

\subsubsection{Finite--dimensional projections}

In the following, we study orthogonal projections of functions in weighted Lebesgue spaces into finite--dimensional subspaces spanned by Hermite and Laguerre polynomials. See \cite{Funaro92} for details of the projection operators.

At first, we consider the weight function $w(x)=\e^{-x^2}$ on the real line $\R$ and denote by $S_M^H$ the vector space spanned by the first $M+1$ Hermite polynomials. The orthogonal projector $\Pi_M^H\colon L^2(\R, w\dx)\to S_M^H$ with
\begin{equation}\label{approx}
\Pi^H_M f=\sum_{i=0}^{M}\frac{\left\langle f,H_i\right\rangle_w}{\left\|H_i\right\|_w^2}H_i=\sum_{i=0}^{M}\frac{\left\langle f,H_i\right\rangle_w}{2^ii!\sqrt{\pi}}H_i \quad\mbox{for }f\in L^2(\R, w\dx)
\end{equation}
satisfies
$$
\left\|f-\Pi^H_M f\right\|_w=\inf_{\phi\in S^H_M}\left\|f-\phi\right\|_w
\quad\mbox{and}\quad
\lim_{M\to\infty}\left\|f-\Pi^H_Mf\right\|_w=0
$$
for every $f\in L^2(\R, w\dx)$. Moreover, for each $k\in\N_0$ there exists a constant $C=C(k)>0$ such that
\begin{equation}\label{eq:H_est}
\left\|f-\Pi^H_M f\right\|_w\leq CM^{-k/2}\left\|\frac{\mathrm{d}^kf}{\mathrm{d}x^k}\right\|_w\quad\mbox{for all }M>k
\end{equation}
and for every $f\in H^k(\R,w\dx)$, see \cite[Theorem 6.2.6]{Funaro92}. We will later use the orthogonal projector $\Pi^H_M$ defined in \eqref{approx} to study the Black-Scholes model.

Next, we consider the weight function $w(v)=\e^{-v}$ on $\R^{+}$ and denote by $S_N^L$ the vector space spanned by the first $N+1$ Laguerre polynomials. The orthogonal projector $\Pi_N^L\colon L^2(\R^{+}, w\dv)\to S_N^L$ with
$$
\Pi^L_N f=\sum_{j=0}^{N}\left\langle f,L_j\right\rangle_wL_j \quad\mbox{for }f\in L^2(\R^{+}, w\dv)
$$
satisfies the same approximation properties as $\Pi^H_M$ and for each $k\in\N_0$ we have
\begin{equation}\label{eq:L_est}
\left\|f-\Pi^L_N f\right\|_w\leq CN^{-k/2}\left\|x^{k/2}\frac{\mathrm{d}^kf}{\mathrm{d}x^k}\right\|_w\quad\mbox{for all }N>k
\end{equation}
for every $f$ with $\frac{\mathrm{d}^mf}{\mathrm{d}x^m}x^{m/2}\in L^2(\R^{+},w\dv)$, $0\leq m\leq k$, and a constant $C=C(k)>0$, see \cite[Theorem 6.2.5]{Funaro92}.

To treat models with non-constant volatility, we use a weighted Lebesgue space in two variables. A \emph{weighted Lebesgue space} $L^2 (\R\times\R^{+}, w\d x\d v)$ with the \emph{weight function} $w: \R\times\R^+ \rightarrow (0,\infty)$ is the space of all measurable functions $g$ for which
\begin{align*}
\Vert g \Vert_{w} := \left( \int_{\R}\int_{\R^+} \vert g(x,v) \vert^2 w(x,v) \d v \d x\right)^{1/2} < \infty.
\end{align*}

The inner product is defined in accordance with \eqref{def:inner}.
For the Heston model, we will consider the weighted Lebesgue space $L^2(\R\times\R^+, w\dx\dv)$ with the weight $w(x,v)=\e^{-x^2-v}$. Due to \cite[Sec. II.4]{ReedSimon80}, the products of Hermite and Laguerre polynomials $P_{i,j}(x,v)=H_i (x) L_j (v)$ for $i,j\in\N_0$ with $\left\langle P_{k,l},P_{i,j}\right\rangle_w = 2^ii!\sqrt{\pi}\cdot\delta_{k,i}\cdot\delta_{l,j}$ for $k,l,i,j\in\N_0$ form a complete orthogonal set in $L^2 (\R\times\R^+, w\dx\dv)$. Let $S_{M,N}$ denote the vector space spanned by the products of the first $M+1$ Hermite polynomials and the first $N+1$ Laguerre polynomials. The orthogonal projector $\Pi_{M,N}\colon L^2(\R\times\R^+, w\dx\dv)\to S_{M,N}$ defined by
\begin{equation}\label{approx2}
\Pi_{M,N}f = \sum_{i=0}^M\sum_{j=0}^N\frac{\left\langle f,P_{i,j}\right\rangle_w}{2^ii!\sqrt{\pi}}P_{i,j}\quad\mbox{for }f\in L^2(\R\times\R^+, w\dx\dv)
\end{equation}
inherits the approximation properties from the projection operators $\Pi^H_M$ and $\Pi^L_N$.

For practical reasons, we have to evaluate the finite-dimensional projections \eqref{approx} and \eqref{approx2} numerically, where the Clenshaw's algorithm 
\citep[Sec. 5.4]{Press07}
will be of use. To evaluate the Fourier coefficients

\begin{equation}\label{eq:fc}
(i)\quad c_i = \frac{\langle f,H_i\rangle_{w}}{2^i i! \sqrt{\pi}}
\qquad\mbox{and}\qquad
(ii)\quad c_{i,j} = \frac{\langle f, P_{i,j}\rangle_{w}}{2^i i! \sqrt{\pi}}
\end{equation}
in \eqref{approx} and \eqref{approx2}
precisely, it is necessary to choose the appropriate quadrature. Here we consider the Gauss--Hermite and Gauss--Laguerre quadratures, see for example in the books
\cite[Sec. 25.4]{AbramowitzStegun64},
\cite[Sec. 14.5 -- 14.7]{Szego75},
\cite[Sec. 3.5]{Olver10} or
\citep[Sec. 4.6]{Press07}.

\subsection{Option pricing models} 
\label{sec23:models}

Since options are frequently traded contracts, the derivation of the option prices is an important task in mathematical finance. There exist several models for option pricing in an arbitrage-free setting. The prices that can be provided by these models give us an idea how the real market prices should behave. We will consider option pricing in the classical models by \cite{BlackScholes73} with constant volatility and by \cite{Heston93} with a mean--reverting stochastic volatility process. 

In this article, we restrict ourselves to the pricing of European call options, since the price of the corresponding European put options can be obtained by the put--call parity. A European option contract is characterized by two parameters, maturity $T$ and strike price $K$. We introduce also a variable $\gamma>0$ that is sometimes called \emph{moneyness} and that measures a relative position of the price $S$ of an underlying asset (typically a stock) with respect to the strike price, i.e. $S=\gamma K$. If $\gamma=1$, we say that the option is at-the-money (ATM), for $\gamma>1$ the call option is in-the-money (ITM) and for $\gamma<1$ it is out-of-the-money (OTM). For put options it is clearly the reverse.

In both models the money market is represented by a risk-free bond
\[ \d B_t = r B_t \d t,\]
with constant interest rate $r>0$.

\subsubsection{Black-Scholes model}

\label{sec2:BS}
Let $(\Omega,\F,\P)$ be a complete probability space with a fixed filtration $(\F_t)$ generated by a standard Wiener process $W_t^S$. In BS model the stock price process $S_t$ is modelled as a continuous semimartingale with respect to $(\F_t)$ and satisfies the stochastic differential equation
\begin{equation}\label{eq:dyn_BS}
\mathrm{d} S_t = \mu S_t\,\mathrm{d}t + \sigma S_t\,\mathrm{d}W_t^S,
\end{equation}
where drift $\mu\in\R$ and volatility $\sigma>0$ are constant.
The fair price $V_t=V(S_t,t)$ of a European call option with maturity $T$ and strike price $K$ is defined by the risk-free pricing formula
\begin{equation}\label{eq:pf_stoch}
V_t = \e^{-r(T-t)} \E^*[(S_T-K)^+ | \F_t ],
\end{equation}
where the conditional expectation is considered under the unique equivalent martingale measure $\P^*$ provided that $V$ is continuous. The equivalent measure $\P^*$ can be obtained from \eqref{eq:dyn_BS} by replacing $\mu$ by $\mu^*=r$ and keep $\sigma^*=\sigma$. It can be shown that $V$ also satisfies the Black-Scholes partial differential equation
\begin{equation*}
\frac{\partial}{\partial t}V + \frac12\sigma^2S^2\frac{\partial^2 V}{\partial S^2} + rS\frac{\partial V}{\partial S} - rV = 0
\end{equation*}
for $(S,t)\in(0,\infty)\times(0,T)$ with the terminal condition $V(S,T)=(S-K)^+$. There exit several approaches to obtain the PDE like replication of the derivative with a self-financing portfolio or delta hedging. For more details on replication strategies we refer to \cite[Chapter 5.8.B]{KaratzasShreve91}. We introduce new variables $\tau=T-t$ and $x=\ln S$, for the time till maturity and the logarithm of the stock price, respectively. For the function $u(x,\tau)=V(S,t)$ we obtain the parabolic Cauchy problem
\begin{equation}
\label{eq:BS}\tag{BS}
\left\{
\begin{alignedat}{2}
\frac{\partial}{\partial\tau}u(x,\tau)
&= \mathcal{L}^{BS}u(x,\tau)
\quad &&\mbox{for }(x,\tau)\in\R\times (0,T),\\
u(x,0) &= (\e^x-K)^+
\quad &&\mbox{for }x\in\R,
\end{alignedat}
\right.
\end{equation}
with the Black-Scholes operator
\begin{equation*}
\mathcal{L}^{\text{BS}}u :=
\frac12\sigma^2\frac{\partial^2}{\partial x^2}u + \left(r-\frac12\sigma^2\right)\frac{\partial}{\partial x}u - ru.
\end{equation*}

\cite{BlackScholes73} formula for the fair price 
of a European call option reads
\begin{align}
\label{BS_formula}
u^{\text{BS}}(x,\tau) &= \e^x N(d_1) - K\e^{-r\tau} N(d_2),
\intertext{where}
d_1 &= \frac{x-\ln{K}+\left(r+\frac12{\sigma^2}\right)\tau}{\sigma\sqrt{\tau}},\nonumber\\
d_2 &= d_1 - \sigma\sqrt{\tau},\nonumber
\end{align}
and $N(\cdot)$ denotes the cumulative distribution function of the standard normal distribution.

\subsubsection{Heston model}

Let $(\Omega,\F,\P)$ be a complete probability space with a filtration $(\F_t)$ and let $W_t^S$ and $W_t^v$ be two standard Wiener processes with respect to the filtration that are correlated by a factor $\rho\in[-1,1]$. In contrast to BS model, in the Heston model the volatility is modelled as the square-root of a mean-reverting stochastic process $v_t$. Both, the stock price process $S_t$ and $v_t$ are continuous semimartingales with respect to $(\F_t)$. The model dynamics are 
\begin{align}\label{eq:dyn_H}
\begin{split}
\mathrm{d} S_t &= \mu S_t\,\mathrm{d}t + \sqrt{v_t} S_t\,\mathrm{d}W_t^S \\
\mathrm{d} v_t &= \kappa(\theta-v_t)\,\mathrm{d}t + \tilde{\sigma}\sqrt{v_t}\,\mathrm{d}W_t^v \\
\rho \mathrm{d}t &= \d W_t^S \d W_t^v.
\end{split}
\end{align}
The drift $\mu\in\R$ and parameter (also called volatility of volatility) $\tilde{\sigma}>0$ are constant. The mean-reverting stochastic variance process $v_t$, also referred to as the \cite{CIR85} process, with constant rate of mean reversion $\kappa$ and long-run mean level $\theta$, both positive, is strictly positive provided that the so called Feller's condition $2\kappa\theta>\tilde{\sigma}^2$ holds. Again, if the function $V_t=V(S_t,v_t,t)$ of the option price is continuous then it is given by the pricing formula \eqref{eq:pf_stoch} for an equivalent martingale measure $\P^*$ that we get from \eqref{eq:dyn_H} by replacing $\mu$, $\kappa$, $\theta$ by $\mu^*=r$, $\kappa^*=\kappa+\lambda>0$, $\theta^*=\kappa\theta/\kappa^*$, respectively, and keep $\tilde{\sigma}^*=\tilde{\sigma}$ and $\rho^*=\rho$. The parameter $\lambda\in[0,\infty)$ is referred to as the price of volatility risk. The price $V$ also satisfies a partial differential equation 
\begin{align*}
\frac{\partial}{\partial t}V + \frac12 vS^2\frac{\partial^2}{\partial S^2}V + \rho\tilde{\sigma} vS\frac{\partial^2}{\partial S\partial v} V  &+ \frac12 \tilde{\sigma}^2  v\frac{\partial^2}{\partial v^2}V \\
&+ rS\frac{\partial}{\partial S}V + [\kappa(\theta-v)-\lambda v]\frac{\partial}{\partial v}V- rV = 0
\end{align*}
for $(S,v,t)\in(0,\infty)^2\times(0,T)$ with the terminal condition $V(S,v,t)=(S-K)^+$ which can be proved for instance with the help of a replicating self-financing portfolio, e.q., one could easily modify the proof in \cite[Section 2.4]{Fouque00}, where an Ohrnstein-Uhlenbeck process is used instead of the CIR process. Without loss of generality we set $\lambda = 0$ by using standard transformation techniques \citep[Sec. 3]{Heston93}.

As above, we introduce the new variables $\tau=T-t$ and $x=\ln S$. For the function $u(x,v,\tau)=V(S,v,t)$ we obtain the initial value problem
\begin{equation}
\label{eq:H}\tag{H}
\left\{
\begin{alignedat}{2}
\frac{\partial}{\partial\tau}u(x,v,\tau)
&= \mathcal{L}^{\text{H}}u(x,v,\tau)
\quad &&\mbox{for }(x,v,\tau)\in\R\times(0,\infty)\times(0,T),\\
u(x,v,0) &= (\e^x-K)^+
\quad &&\mbox{for }(x,v)\in\R\times(0,\infty),
\end{alignedat}
\right.
\end{equation}
with the partial differential operator
\begin{equation*}
\mathcal{L}^{\text{H}}u :=
\frac12v\frac{\partial^2u}{\partial x^2}
+ \rho\tilde{\sigma} v\frac{\partial^2u}{\partial x\partial v}
+ \frac12\tilde{\sigma}^2 v\frac{\partial^2u}{\partial v^2}
+ \left(r-\frac12v\right)\frac{\partial u}{\partial x}
+ \kappa(\theta-v) \frac{\partial u}{\partial v}
- ru.
\end{equation*}

In the book by \cite{Lewis00}, the author presents the so called fundamental transform approach for solution of the initial value problem \eqref{eq:H}. We present here only the pricing formula that has among others one numerical advantage in the sense that we have to calculate only one numerical integral for each price of the option (compared to the two-integrals formula by Heston). The price of the European call option can be expressed as the so called Heston-Lewis formula

\begin{align}\label{e:lewis}
u^{\text{H}}(x,v,\tau) &= \e^x - K~\e^{-r\tau} \frac1{\pi} \int_{0+i/2}^{+\infty+i/2} {e^{-ikX} \frac{\hat H(k,v,\tau)}{k^2-ik}} \d k,
\end{align}
where $X= x - \ln(K) + r\tau$ and
\begin{align*}
\hat H(k,v,\tau) =& \exp\Bigg( \frac{2\kappa\theta}{\tilde{\sigma}^2} \bigg[ q\,g - \ln\bigg(\frac{1-he^{-\xi q}}{1-h} \bigg)\bigg]
+ vg \bigg(\frac{1-e^{-\xi q}}{1-he^{-\xi q}} \bigg) \Bigg),
\end{align*}
where
\begin{align*}
g &= \frac{b-\xi}{2},\quad h=\frac{b-\xi}{b+\xi},\quad q=\frac{\tilde{\sigma}^2 \tau}{2}, \\
\xi &= \sqrt{ b^2 + \frac{4(k^2 - ik)}{\tilde{\sigma}^2} }, \\
b &= \frac{2}{\tilde{\sigma}^2} \Big(ik\rho\tilde{\sigma} + \kappa\Big).
\end{align*}

To show that the original \cite{Heston93} pricing formula and \eqref{e:lewis} are equivalent, we refer to the paper by \cite{BaustianMrazekPospisilSobotka17asmb}, where the authors also extended Lewis's approach to models with jumps. 

\section{Methodology}\label{sec:methodology}

In this section we present our main results, in particular we introduce our Galerkin-based method. First, we establish the weak formulation of the Black-Scholes equation in a weighted Lebesgue space and show how we can solve the equation in finite-dimensional subspaces spanned by Hermite polynomials. The smooth solutions in the finite-dimensional subspaces approximate the weak solution of the Black-Scholes equation. Although the Black-Scholes model has already been studied in detail, Section~\ref{sec:BSPDE} gives us a good understanding how the method should work for the more complicated Heston model. Second, we establish the method for the Heston model and study the equation for vanishing volatility.

As we have pointed out in the introduction, the Galerkin method for parabolic equations and their convergence properties were widely studied in the past. Even so, our applications are special in the sense that we have an unbounded domain and unbounded initial data. Most numerical schemes for unbounded domains just cut the domain at a certain point. Contrary to this, we use an orthogonal base on the whole unbounded domain. To treat the unbounded initial condition we consider weighted Lebesgue spaces.

\subsection{Solution of the Black-Scholes PDE}\label{sec:BSPDE}

Let us now consider the parabolic Cauchy problem for the function $u(x,\tau)=V(S,t)$ introduced in Section \ref{sec2:BS}
\begin{align}
\label{BS:PDE}
\left\{
\begin{alignedat}{2}
\frac{\partial}{\partial\tau}u(x,\tau)
&= \mathcal{L}^{\text{BS}}u(x,\tau)
\quad &&\mbox{for }(x,\tau)\in\R\times (0,T),\\
u(x,0) &= (\e^x-K)^+
\quad &&\mbox{for }x\in\R,
\end{alignedat}
\right.
\end{align}
with the Black-Scholes operator
\begin{equation*}
\mathcal{L}^{\text{BS}}u :=
\frac12\sigma^2\frac{\partial^2}{\partial x^2}u + \left(r-\frac12\sigma^2\right)\frac{\partial}{\partial x}u - ru.
\end{equation*}
The initial data is obviously not in $L^2(\R)$ but in the weighted Lebesgue space $L^2(\R,w\dx)$ with the weight function $w(x)=\e^{-x^2}$ and even in the weighted Sobolev space $H^1(\R,w\dx)$.
We want to obtain a weak formulation of the problem in the weighted space. Therefore, we multiply the partial differential equation \eqref{BS:PDE} with a test function $\phi \in C^{\infty}_0(\R)$ and the weight function $w$. If we integrate over $\R$
then integration by parts yields
\begin{equation*}
\int_{-\infty}^{\infty}\frac{\partial u}{\partial\tau}\phi w\dx
+
\int_{-\infty}^{\infty} \left[ \frac12 \sigma^2 \frac{\partial u}{\partial x} \frac{\partial \phi}{\partial x}  - \left( r - \frac12 \sigma^2 +\sigma^2 x\right) \frac{\partial u}{\partial x} \phi + r u \phi\right] w \dx = 0.
\end{equation*}
Following the standard procedure described for example in \cite[p. 296]{Evans10}, 
we define the bilinear form
\begin{equation}
\label{BS:B}
\B(\varphi,\psi) := \int_{-\infty}^{\infty} \left[ \frac12\sigma^2 \frac{\partial \varphi}{\partial x}\frac{\partial \psi}{\partial x} + \left(\frac12\sigma^2-r-\sigma^2x\right) \frac{\partial \varphi}{\partial x}\psi + r\varphi\psi \right] w \dx
\end{equation}
for $\varphi,\psi\in H^1(\R,w\dx)$. We call
$$
u\in L^2((0,T)\to H^1(\R,w\dx)) \quad\mbox{with}\quad \frac{\mathrm{d}}{\mathrm{d}\tau}u\in L^2((0,T)\to H^{-1}(\R,w\dx))
$$
a weak solution of \eqref{BS:PDE} if 
\begin{equation}\label{BS:weak}
\left\langle \frac{\d}{\d \tau} u, \phi \right\rangle_w + \B(u, \phi) = 0
\end{equation}
for each test function $\phi \in H^1(\R, w\dx)$ and a.e. time $0\leq\tau\leq T$, and $u(0)=(\e^{x}-K)^{+}$. Here, $H^{-1}$ is the dual space of the Sobolev space $H^1$ and can be canonically identified with it by the Riesz representation theorem. The existence of the unique weak solution in the weighted space can be obtained by modifying the proof of the Galerkin method in \cite[Chapter 7.1, p. 349]{Evans10}.

Following the Galerkin method, we want to approximate the weak solution $u$ with solutions $u_M$ of the Cauchy problem \eqref{BS:PDE} in the finite--dimensional subspace $S^H_M$, i.e., we look for a solution $u_M$ in the form
\begin{equation}
u_M (x, \tau) = \sum_{k=0}^{M} c_k (\tau) H_k (x) \label{BS:sol}
\end{equation}
with a given initial condition
\begin{equation}
u_M (x,0) = \sum_{k=0}^{M} c_k (0) H_k (x), \label{BS:ic}
\end{equation}
where $c(\tau)$ is a column vector of Fourier coefficients $c(\tau) = \left[ c_0 (\tau), c_1 (\tau), \ldots , c_M (\tau) \right]^\text{T}$, where $\text{T}$ denotes the transposition (not to be confused with time $T$). 

The natural choice for the initial condition is the orthogonal projection of the payoff function $\Pi^H_Mu(x,0)$ defined in \eqref{approx}. For instance, the coefficients in the initial condition \eqref{BS:ic} satisfy
\begin{equation}\label{BS:ic:coef}
c_k(0) = c_{0,k} := \int_{-\infty}^{+\infty} (\e^x - K)^+ H_k(x) \e^{-x^2} \d x, \quad k=0,1,\dots,M,
\end{equation}
or in vector form $c(0) = c_0 = [c_{0,0}, c_{0,1}, \dots, c_{0,M}]^\text{T}$.

Let us now substitute $\varphi = H_i (x), \psi = H_j (x)$ into the bilinear form \eqref{BS:B}. 
In the view of Lemma~\ref{lem25} we can simplify the term $\B(H_i,H_j)$ and obtain the explicit form
\begin{equation}
\frac{1}{2^j j! \sqrt{\pi}}\B(H_i,H_j)
= i (\sigma^2 - 2r) \delta_{i,j+1} - 2 i \sigma^2 (j+1) \delta_{i,j+2} + r \delta_{i,j}.  \label{BS_matrixB}
\end{equation}

We plug \eqref{BS:sol} into \eqref{BS:weak} and choose the Hermite polynomial $H_j$ as the test function
\begin{equation*}
\left\langle \sum_{k=0}^{M} c_k' (\tau) H_k(x), H_j(x) \right\rangle_w + \sum_{k=0}^{M} c_k (\tau) \B(H_k(x),H_j(x)) = 0.
\end{equation*}
We make use of the orthogonality of the Hermite polynomials to obtain a system of ODEs
\begin{equation}
c_j' (\tau) + \frac{1}{2^j j! \sqrt{\pi}} \sum_{k=0}^{M} c_k(\tau) \B(H_k(x), H_j(x)) = 0, \quad j=0,1,...,M, \label{BS:ODEs_indiv}
\end{equation}
that possesses a unique solution to the initial data \eqref{BS:ic:coef} by standard existence theory.

Let us introduce a matrix $B = \left[ B_{k,j} \right]$, $k,j = 0, 1, \ldots, M$,
with elements
\begin{equation}\label{BS_matrixBdef}
B_{k,j}:= \frac{1}{2^j j! \sqrt{\pi}} \B(H_k, H_j)
\end{equation}
and denote by $B^{T}$ the transposed matrix\footnote{In the implementation, one can easily swap the arguments of the bilinear form in \eqref{BS_matrixBdef} in order to get an already transposed matrix. However, in the text we prefer the \emph{natural ordering} and hence the transposition in the formulas below is needed.}. From \eqref{BS_matrixB} we can easily see that $B^\text{T}$ is a three-diagonal matrix with entries on the main diagonal and two superdiagonals.
With the matrix $B^\text{T}$ we can rewrite \eqref{BS:ODEs_indiv} in the matrix form as
\begin{equation}
\frac{\d }{\d \tau} c(\tau) + B^\text{T} c(\tau) = 0, \qquad c(0) = c_0. \label{BS:ODEs}
\end{equation}
We can write the solution in terms of the matrix exponential as
\begin{equation}\label{BS:expm}
c(\tau) = \e^{-B^\text{T} \tau} c_0.
\end{equation}

\subsection{Solution of the Heston PDE}\label{sec:HestonPDE}

Let us now consider the \cite{Heston93} model with stochastic volatility. As above, we can use the Hermite polynomials for the polynomial expansion in the variable connected to the logarithm of the stock price. However, for the volatility variable we prefer Laguerre polynomials due to the fact that the volatility is strictly positive. The Cauchy problem connected to the model of \cite{Heston93} is
\begin{equation}
\label{H:PDE}
\left\{
\begin{alignedat}{2}
\frac{\partial}{\partial\tau}u(x,v,\tau)
&= \mathcal{L}^{\text{H}}u(x,v,\tau)
\quad &&\mbox{for }(x,v,\tau)\in \R \times (0,\infty) \times (0,T),\\
u(x,v,0) &= (\e^x-K)^+
\quad &&\mbox{for }x\in \R \times (0,\infty).
\end{alignedat}
\right.
\end{equation}
with the Heston operator
\begin{equation*}
\mathcal{L}^{\text{H}}u :=
\frac12 v \frac{\partial^2 u}{\partial x^2} + \rho \tilde{\sigma} v \frac{\partial^2 u}{\partial x \partial v} + \frac12 \tilde{\sigma}^2 v \frac{\partial^2 u}{\partial v^2} \\
+ \left( r - \frac12 v \right) \frac{\partial u}{\partial x} + \kappa(\theta-v) \frac{\partial u}{\partial v} - ru.
\end{equation*}

To obtain a weak formulation of the solution we multiply \eqref{H:PDE} with a test function $\phi \in C^{\infty}_0(\R \times \R^+)$ and the weight function $w(x,v)=\e^{-x^2-v}$. Integration over the domain $\R \times (0, \infty)$ and application of Gauss's theorem then yields the variational formulation of the problem
\begin{equation*}
\int_{0}^{+\infty} \int_{-\infty}^{+\infty} \frac{\partial u}{\partial \tau} \phi w\dx\dv + \tilde{\B}(u,\phi) = 0
\end{equation*}
with the bilinear form $\tilde{\B}$ defined by
\begin{align}
\begin{split}\label{H:B}
\tilde{\B}(\varphi, \psi) &= \int_{0}^{+\infty} \int_{-\infty}^{+\infty} \frac12 v \frac{\partial \varphi}{\partial x} \frac{\partial \psi}{\partial x} w \dx \d v + \int_{0}^{+\infty} \int_{-\infty}^{+\infty} \frac12 \rho \tilde{\sigma} v \frac{\partial \varphi}{\partial v} \frac{\partial \psi}{\partial x} w \dx \d v \\
& + \int_{0}^{+\infty} \int_{-\infty}^{+\infty} \frac12 \rho \tilde{\sigma} v \frac{\partial \varphi}{\partial x} \frac{\partial \psi}{\partial v} w \dx \d v
 + \int_{0}^{+\infty} \int_{-\infty}^{+\infty} \frac12 \tilde{\sigma}^2 v \frac{\partial \varphi}{\partial v} \frac{\partial \psi}{\partial v} w \dx \d v \\
& + \int_{0}^{+\infty} \int_{-\infty}^{+\infty} \left( -xv - \frac12 \rho \tilde{\sigma} v + \frac12 \rho \tilde{\sigma} -r +\frac12 v\right) \frac{\partial \varphi}{\partial x} \psi w \dx \d v \\
& + \int_{0}^{+\infty} \int_{-\infty}^{+\infty} \left( -xv\rho \tilde{\sigma} -\frac12\tilde{\sigma}^2 v + \frac12 \tilde{\sigma}^2 -\kappa (\theta - v) \right) \frac{\partial \varphi}{\partial v} \psi w \dx \d v \\
& + \int_{0}^{+\infty} \int_{-\infty}^{+\infty} r \varphi \psi w \dx \d v
\end{split}
\end{align}
for all $\varphi,\psi\in H^1(\R\times\R^{+},w\dx\dv)$.

Similarly as for the BS model, we substitute the elements of the complete orthogonal set $\varphi = P_{i,j} (x,v) = H_i (x) L_j (v)$ and $\psi = P_{k,l} (x,v)= H_k (x) L_l (v)$ into the bilinear form \eqref{H:B}. For the purpose of better clarity, we study all seven integral terms separately. In particular, let
\begin{equation}\label{h:matrixB}
\frac{1}{2^k k! \sqrt{\pi}} \tilde{\B}(P_{i,j}(x,v),P_{k,l}(x,v)) := \frac{1}{2^k k! \sqrt{\pi}} \sum\limits_{r=1}^7 \tilde{\B}_r(P_{i,j}(x,v),P_{k,l}(x,v)),
\end{equation}
where each $\tilde{\B}_r(P_{i,j}(x,v),P_{k,l}(x,v))$, $r$ = 1, 2, ..., 7, represents individual integral terms. 

\begin{theorem}\label{t:int}
The integrals in \eqref{h:matrixB} satisfy
\begin{align*}
\frac{1}{2^k k! \sqrt{\pi}} & \tilde{\B}_1(P_{i,j},P_{k,l})
= i \delta_{i,k} ((2j+1) \delta_{j,l} - j \delta_{j-1,l} - (j+1) \delta_{j+1,l}), \\
\frac{1}{2^k k! \sqrt{\pi}} & \tilde{\B}_2(P_{i,j},P_{k,l})
= \frac12 \rho \tilde{\sigma} \delta_{i+1,k} j (\delta_{j,l} - \delta_{j-1,l}), \\
\frac{1}{2^k k! \sqrt{\pi}} & \tilde{\B}_3(P_{i,j},P_{k,l}) 
= i \rho \tilde{\sigma} \delta_{i,k+1} l (\delta_{j,l} - \delta_{j,l-1}), \\
\frac{1}{2^k k! \sqrt{\pi}} & \tilde{\B}_4(P_{i,j},P_{k,l}) 
= \frac12 \tilde{\sigma}^2 j \delta_{i,k} \delta_{j,l}, \\
\frac{1}{2^k k! \sqrt{\pi}} & \tilde{\B}_5(P_{i,j},P_{k,l})
= 2i \left( \frac12 \rho \tilde{\sigma} - r \right) \delta_{i,k+1} \delta_{j,l}, \\
&+ [-2i(k+1) \delta_{i,k+2}-i \delta_{i,k}][(2j+1) \delta_{j,l} - j \delta_{j-1,l} - (j+1) \delta_{j+1,l}] \\
&+ [i (1-\rho \tilde{\sigma} ) \delta_{i,k+1}][(2j+1) \delta_{j,l} - j \delta_{j-1,l} - (j+1) \delta_{j+1,l}] \\
\frac{1}{2^k k! \sqrt{\pi}} & \tilde{\B}_6(P_{i,j},P_{k,l}) 
= \left( \frac12 \tilde{\sigma}^2 -\kappa \theta \right) \delta_{i,k} \left( - \sum_{a=0}^{j-1} \delta_{a,l} \right) \\
&+ \left[- \rho \tilde{\sigma} \left( \frac12 \delta_{i+1,k} + i \delta_{i-1,k} \right)+ \left( \kappa - \frac12 \tilde{\sigma}^2 \right) \delta_{i,k}\right] j (\delta_{j,l} - \delta_{j-1,l} ), \\
\frac{1}{2^k k! \sqrt{\pi}} & \tilde{\B}_7(P_{i,j},P_{k,l}) 
= r \delta_{i,k} \delta_{j,l},
\end{align*}
for all $0\leq i,k\leq M$ and all $0\leq j,l\leq N$.
\end{theorem}

\begin{proof}
For the calculation of $\tilde{\B}_1$ we apply \eqref{IH1} and \eqref{IL1}. For $\tilde{\B}_2$ we use \eqref{IH2} and \eqref{IL2}. $\tilde{\B}_3$ is derived with the help of a modification of \eqref{IH2}
\begin{equation}\label{IH3}
\frac{1}{2^n n! \sqrt{\pi}} \int_{-\infty}^{\infty} H_m' (x) H_n (x)\e^{-x^2} \d x = 2m \delta_{m,n+1}.
\end{equation}
and \eqref{IL2}. For $\tilde{\B}_4$ we need \eqref{IL3}. In the calculation of $\tilde{\B}_5$ we make use of \eqref{IH4}, \eqref{IL1}, and \eqref{IH3}. For $\tilde{\B}_6$ we need the same equations as for $\tilde{\B}_5$ and \eqref{IL4}. $\tilde{\B}_7$ is trivial. More detailed calculations can be found in the thesis \citep[Sec. 3.2]{Filipova19}.
\end{proof}

In analogy to the BS case, we say that $u\in L^2((0,T)\to H^1(\R\times\R^{+},w\dx\dv))$ with $\frac{\mathrm{d}}{\mathrm{d}\tau}u\in L^2((0,T)\to H^{-1}(\R\times\R_0^{+},w\dx\dv))$ is a weak solution of \eqref{H:PDE} if

\begin{equation}\label{H:weak}
\left\langle \frac{\d}{\d \tau} u, \phi \right\rangle_w + \tilde{\B}(u, \phi) = 0
\end{equation}
for each test function $\phi \in H^1(\R\times\R^{+}, w\dx\dv)$ and a.e. time $0\leq\tau\leq T$, and $u(0,v)=(\e^{x}-K)^{+}$ for all $v>0$. The existence and uniqueness of the weak solution is given by the Galerkin Method in \cite[Chapter 7.1, p. 349]{Evans10}. A detailed proof of the existence of the unique solution in a convenient weighted space considering the boundary conditions can be found in \cite{Alziary18}.

We study solutions of the Cauchy problem \eqref{H:PDE} in finite--dimensional subspaces $S_{M,N}$, i.e. we look for the solution $u$ in the form
\begin{equation}
u_{M,N} (x,v,\tau) = \sum_{i=0}^{M} \sum_{j=0}^{N} c_{i,j} (\tau) P_{i,j} (x,v), \label{eq:solH}
\end{equation}
where
$$
P_{i,j} (x,v) = H_i (x) L_j (v),\quad i, j \in \N_0
$$
and $c_{i,j}(\tau)$, $i=0,1,\dots,M$; $j=0,1,\dots,N$; are (yet unknown) Fourier coefficients.

Let $c(\tau) = [c_a(\tau)]^\text{T}$, $a=0,1,\dots,(M+1)(N+1)$, be a column vector of these coefficients, where $a=i(N+1)+j$, $i=0,1,\dots,M$; $j=0,1,\dots,N$; i.e.

$$c(\tau) = [c_{0,0}(\tau),\dots,c_{0,N}(\tau),c_{1,0}(\tau),\dots,c_{1,N}(\tau),\dots,c_{M,0}(\tau),\dots,c_{M,N}(\tau)]^\text{T}.$$

For the initial data we choose the orthogonal projection of the payoff function $\Pi_{M,N}u(0,v)$, where for $i=0,1,\dots,M$, $j=0,1,\dots,N$
\begin{equation}\label{H:system_odes_ic}
c_{i,j}(0) = c_{0,i,j} := \int_{0}^{+\infty} \int_{-\infty}^{+\infty} (\e^x - K)^+ P_{i,j} (x,v) \e^{-x^2} \e^{-v} \dx \d v,
\end{equation}
or in vector form $c(0) = c_0 = [c_{0,0,0},\dots,c_{0,0,N},c_{0,1,0},\dots,c_{0,1,N},\dots,c_{0,M,0},\dots,c_{0,M,N}]^\text{T}$.

We use \eqref{H:weak} with $u_{M,N}$ of the form \eqref{eq:solH} and the test function $P_{k,l}$. Thanks to the orthogonality of the polynomials we obtain
\begin{equation*}
c_{k,l}'(\tau) + \sum_{i=0}^{M} \sum_{j=0}^{N} c_{i,j} (\tau) \frac{1}{2^k k! \sqrt{\pi}} \tilde{\B}(P_{i,j} (x,v), P_{k,l} (x,v)) = 0.
\end{equation*}

Let us introduce a matrix $\tilde{B} = [\tilde{B}_{a,b}]$, $a,b=0,1,\dots,(M+1)(N+1)$ defined as
\begin{equation}\label{H:Bab}
\tilde{B}_{a,b} = \frac{1}{2^k k! \sqrt{\pi}} (\tilde{\B}(P_{i,j} (x,v),P_{k,l} (x,v))) ,
\end{equation}
where $a = i(N+1) + j; b = k(N+1)+l; i,k=0,\ldots,M; j,l=0,\ldots,N$. Using this assembly\footnote{Swapping the arguments in the bilinear form in \eqref{H:Bab} can again easily produce an already transposed matrix.} it can be shown (by using Theorem \ref{t:int}) that the transposed matrix $\tilde{B}^\text{T}$ is an upper triangular matrix with elements on the main diagonal and $2N+3$ superdiagonals if $N>0$ and 2 superdiagonals in the degenerate case $N=0$ like in the BS case. It is worth to mention that the BS PDE is not a special case of the Heston PDE. The superdiagonal $2N+3$ is a contribution of the term $\tilde{\B}_5$, whose elements lie on seven superdiagonals ($1,N,N+1,N+2,2N+1,2N+2,2N+3$).

As above, we obtain a system of ODEs
\begin{equation}
\frac{\d}{\d \tau} c(\tau) + \tilde{B}^\text{T} c(\tau) = 0, \qquad c(0) = c_0. \label{H:system_odes}
\end{equation}
The solution can also be written in terms of the matrix exponential as
\begin{equation}\label{H:expm}
c(\tau) = \e^{-\tilde{B}^\text{T} \tau} c_0.
\end{equation}

\subsubsection{Solution behaviour analysis near $v = 0$}\label{sec:transport}

We are interested in the behaviour of the solution of the Heston PDE for small volatility, especially at the boundary $v=0$. Motivated by \cite{Alziary20pre}, we study the partial differential equation for $v \rightarrow 0+$.

The solution $u = u(x,v,\tau)$ satisfies the Heston PDE
\begin{equation*}
\frac{\partial}{\partial\tau}u =
\frac12 v \frac{\partial^2 u}{\partial x^2} + \rho \tilde{\sigma} v \frac{\partial^2 u}{\partial x \partial v} + \frac12 \tilde{\sigma}^2 v \frac{\partial^2 u}{\partial v^2} \\
+ \left( r - \frac12 v \right) \frac{\partial u}{\partial x} + \kappa(\theta-v) \frac{\partial u}{\partial v} - ru
\end{equation*} 
in $\R\times\R_0^{+}\times(0,T)$ and can be rewritten as
\begin{equation*}
\frac{\partial u}{\partial \tau} = v \left( \frac12 \frac{\partial^2 u}{\partial x^2} + \rho\tilde{\sigma}\frac{\partial^2 u}{\partial x \partial v} +\frac12 \tilde{\sigma}^2 \frac{\partial^2 u}{\partial v^2} \right) + \left( r - \frac12 v\right) \frac{\partial u}{\partial x} + \kappa(\theta - v) \frac{\partial u}{\partial v} -ru.
\end{equation*}
For $v \rightarrow 0+$ the equation degenerates to the first order equation as shown in \cite[cor. 4.3]{Alziary20pre}
\begin{align*}
\frac{\partial u}{\partial \tau} = r \frac{\partial u}{\partial x} + \kappa\theta \frac{\partial u}{\partial v} - ru.
\end{align*}
Since we want to study the problem for vanishing volatility, we replace the derivative with respect to $v$ by the differential quotient $\frac{1}{h}(u(x,h,\tau)-u(x,0,\tau))$, where $h>0$ denotes a small distance to the boundary. By doing this, we obtain an initial value problem on the boundary
\begin{align}
\label{eq:bound}
\left\{
\begin{alignedat}{2}
\mathcal{L}^{\text{B}}u(x,0,\tau)
&= \frac{\kappa\theta}{h} u(x,h,\tau)
\quad &&\mbox{for }(x,\tau)\in\R\times (0,T),\\
u(x,0,0) &= (\e^x-K)^+
\quad &&\mbox{for }x\in\R,
\end{alignedat}
\right.
\end{align}
with the unknown function $u(x,0,\tau)$ for fixed volatility $v=0$ and with the differential operator
$$
\mathcal{L}^{\text{B}}u=\frac{\partial u}{\partial \tau} - r \frac{\partial u}{\partial x} + \left(r+\frac{\kappa\theta}{h}\right)u.
$$
We can derive a solution of the Cauchy problem in dependence of the inhomogeneity which consists of values of the solution of the Heston equation away from the boundary.

We introduce a new variable $y=x+r\tau$ and the function $\tilde{u}(y,\tau)=u(x,0,\tau)$ that satisfies the inhomogeneous \emph{transport equation}
\begin{align*}
\frac{\partial\tilde{u}}{\partial \tau} + \left(\frac{\kappa\theta}{h} + r \right)\tilde{u} = \frac{\kappa\theta}{h}u(y-r\tau,h,\tau)
\end{align*}
with the initial condition $\tilde{u}(y,0)=(\e^{y-r\tau}-K)^+$. Following the standard procedure, we define the function
$$
U(y,\tau)=\e^{(\frac{\kappa\theta}{h}+r)\tau}\tilde{u}(y,\tau)
$$
and obtain
\begin{equation}
\label{e:transport}
\frac{\partial U}{\partial \tau}
= \frac{\kappa\theta}{h}\e^{(\frac{\kappa\theta}{h}+r)\tau}u(y-r\tau,h,\tau) \quad\mbox{with}\quad
U(y,0) = (\e^{y-r\tau} -K)^+.
\end{equation}
We integrate \eqref{e:transport} with respect to the time variable
\begin{equation*}
U(y,\tau) =(\e^{y-r\tau} -K)^+ + \frac{\kappa\theta}{h}\int_{0}^{\tau} \e^{(\frac{\kappa\theta}{h}+r)\xi}u(y-r\xi,h,\xi) \d \xi.
\end{equation*}
Hence, we get the \emph{boundary solution} that we denote as
\begin{equation*}
u^{\text{B}}(x,0,\tau) = \e^{-(\frac{\kappa\theta}{h}+r)\tau}\left[ (\e^x-K)^+ + \frac{\kappa\theta}{h} \int_{0}^{\tau} \e^{(\frac{\kappa\theta}{h}+r)\xi}u(x+r(\tau-\xi),h,\xi) \d \xi \right]
\end{equation*}
and in particular
\begin{equation}\label{e:transport_sol}
u^{\text{B}}(x,0,T) = \e^{-(\frac{\kappa\theta}{h}+r)T}\left[ (\e^x-K)^+ + \frac{\kappa\theta}{h} \int_{0}^{T} \e^{(\frac{\kappa\theta}{h}+r)\xi}u(x+r(T-\xi),h,\xi) \d \xi \right].
\end{equation}
These formulas contain an integral over the finite interval $[0,T]$. For the values of $u$ for $h>0$ we could make use of the polynomial expansion of the solution obtained in Section~\ref{sec:HestonPDE}. 

\section{Results}\label{sec:results}

In this section we present numerical results for several particular examples. All supporting codes are implemented in MATLAB. Parameter values in considered examples are chosen consistently with other cited resources in order to demonstrate the functionality of the proposed method. To provide a thorough analysis of the numerical solution for all possible parameter values combinations goes beyond the scope of present paper. When we refer to the $L^2$ error it is the error with respect to the norm of the weighted Lebesgue spaces $L^2(\R,\mathrm{e}^{-x^2}\dx)$ and $L^2(\R\times\R^{+},\mathrm{e}^{-x^2-v}\dx\dv)$, respectively. For convenience, point-wise error is calculated for several selected nodes as well as the average absolute and relative error. We compare the newly proposed solution to existing closed formula \eqref{BS_formula} for BS model and semi-closed formula \eqref{e:lewis} for Heston model, respectively.

\subsection{Black-Scholes model}\label{ssec:results:BS}

In the following setting for BS model the parameters are chosen as follows:
\begin{itemize}
\item volatility $\sigma = 0.03$,
\item risk free interest rate $r= 0.1$,
\end{itemize}
and options parameters are the following:
\begin{itemize}
\label{options_param}
\item maturity $T = 1$,
\item strike price $K = 100$,
\item stock price $S \in [0;2K]$ discretized with the equidistant step $\Delta S = 0.01$,
\end{itemize} 
and we impose $x = \ln(S)$ (with $\ln(\Delta S)$ being the smallest discretization point in $x$ variable). In the case of BS model, we choose Hermite polynomials as the complete orthogonal system of polynomials and focus on solving the Black-Scholes PDE.
Our numerical solution $u_M$ of the BS PDE \eqref{BS:PDE} is considered in the form \eqref{BS:sol}.
Fourier coefficients for $\tau = T$ are obtained by solving the system of ODEs~\eqref{BS:ODEs} with $B$ given by \eqref{BS_matrixB} and \eqref{BS_matrixBdef}. We use MATLAB ODE solver \texttt{ode45} to solve this system since it leads to smaller values of $L^2$~errors then the numerical calculation of matrix exponential in \eqref{BS:expm} using MATLAB procedure \texttt{expm}. See \citep[Sec.~4.1]{Filipova19} for a comparison of these two methods. 

The BS formula $u^{\text{BS}}$ and solutions $u_M$ obtained by \texttt{ode45} for $M=20$ and $M=120$ are shown in Figure~\ref{fig:pde_xS} on the left. For convenience the solution is plotted only for $x\geq 0$. The three vertical dashed lines represent moneyness $\gamma\in\{0.7, 1, 1.3\}$ respectively. On the right we can see the behaviour of the absolute error in this region.

\begin{table}[ht]
\centering
\caption{Errors comparison for BS model for different Hermite polynomial orders: absolute error AE and relative error RE at several selected nodes are listed.}\label{tab:BS_AEREerrors}
\begin{tabular}{lllllll}
\toprule
$M$ & $\AE(0.7)$ & $\RE(0.7)$ & $\AE(1)$ & $\RE(1)$ & $\AE(1.3)$ & $\RE(1.3)$ \\
\midrule
20 & 0.22055 & 185.605 & 0.152195 & 0.027266 & 8.00953 & 0.243 \\
40 & 0.0991825 & 83.4675 & 0.287501 & 0.0515061 & 1.38974 & 0.0421632 \\
60 & 0.0875455 & 73.6744 & 0.249526 & 0.0447028 & 0.196272 & 0.00595466 \\
80 & 0.0476676 & 40.1149 & 0.225168 & 0.0403392 & 0.542587 & 0.0164615 \\
100 & 0.0209973 & 17.6704 & 0.229207 & 0.0410627 & 0.456833 & 0.0138598 \\
120 & 0.00662291 & 5.57354 & 0.230279 & 0.0412547 & 0.312095 & 0.00946862 \\
\bottomrule
\end{tabular}
\end{table}

\begin{figure}[h]
\centering
\includegraphics[width=\textwidth,trim={22mm 0 29mm 0},clip]{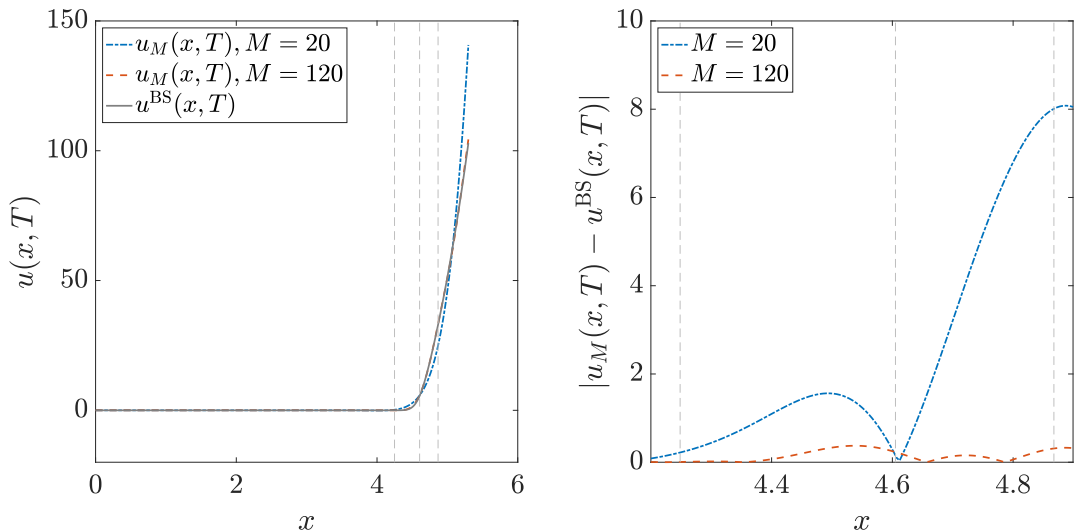}
\caption{Solution $u_M$ of the BS PDE for $M = 20$ and $M = 120$ together with the BS formula $u^{\text{BS}}$ is depicted on the left and the absolute error on the right. Vertical grid lines plotted at $\gamma\in\{0.7, 1, 1.3\}$.}
\label{fig:pde_xS}
\end{figure}

\begin{table}[ht]
\centering
\caption{Errors comparison for BS model for different Hermite polynomial orders: $L^2$ error, average absolute error AAE and average relative error ARE are listed for two sets of points 1 and 2.}\label{tab:BS_L2errors}
\begin{tabular}{llllll}
\toprule
$M$ & $L^2$ error & $\AAE_{(1)}$ & $\ARE_{(1)}$ & $\AAE_{(2)}$ & $\ARE_{(2)}$ \\
\midrule
20 & 3.11957e-09 & 2.81507 & 11.9282 & 5.29885 & 0.190209 \\
40 & 5.23407e-10 & 1.00759 & 3.23288 & 1.57162 & 0.0630167 \\
60 & 1.65443e-10 & 0.451035 & 3.10148 & 1.01331 & 0.0370747 \\
80 & 7.21744e-11 & 0.274018 & 1.91442 & 0.575191 & 0.0220953 \\
100 & 4.05356e-11 & 0.196393 & 1.02488 & 0.277925 & 0.0127912 \\
120 & 2.93889e-11 & 0.152986 & 0.473851 & 0.17448 & 0.00880548 \\
\bottomrule
\end{tabular}
\end{table}

\begin{figure}[h]
\centering
\includegraphics[width=\textwidth,trim={22mm 0 29mm 0},clip]{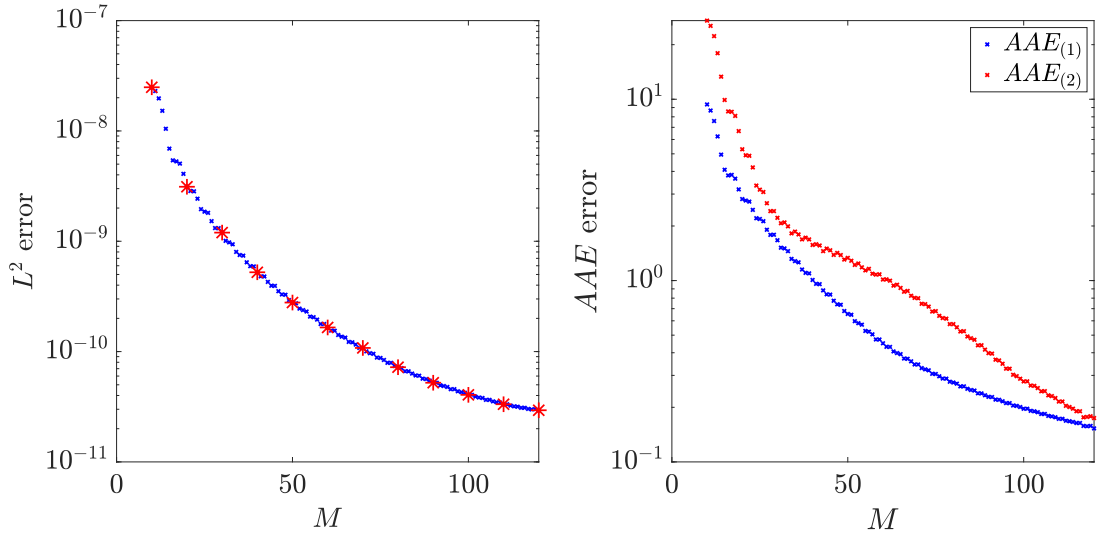}
\caption{Convergence of the $L^2$~error on the left. Red asterisks indicate values that are listed in the second column of Table~\ref{tab:BS_L2errors}. On the right we can see a convergence of the average absolute error calculated for two different sets of nodes. The scale at the vertical axis is logarithmic in both pictures.}
\label{fig:BS_convergence}
\end{figure}

We measure the following errors. First and foremost we compute the $L^2$~error using the Gauss-Hermite quadrature with 251 Hermite points. In the second column of Table~\ref{tab:BS_L2errors} we list the corresponding $L^2$~error for different polynomial orders. Convergence of the $L^2$~error is visually depicted in Figure~\ref{fig:BS_convergence} on the left. The set of nodes used on the right of Figure~\ref{fig:BS_convergence} and in Table~\ref{tab:BS_L2errors} will be introduced below.

Next we measure the point-wise absolute and relative errors at selected nodes and their average. In particular, by $\AE(\gamma)$ we denote the absolute error with respect to the Black-Scholes formula \eqref{BS_formula} at the point $S=\gamma K$,
\[ \AE(\gamma) = |u_M(\ln(\gamma K),T)) -  u^{\text{BS}}(\ln(\gamma K),T)|,  \]
where $\gamma>0$ is the moneyness introduced in Section \ref{sec23:models}. Similarly we measure the relative error
\[ \RE(\gamma) = \left|1-\frac{u_M(\ln(\gamma K),T))}{u^{\text{BS}}(\ln(\gamma K),T)}\right|. \]
In Table~\ref{tab:BS_AEREerrors}, we list the values of both absolute and relative error at three different moneyness nodes ($\gamma\in\{0.7, 1, 1.3\}$) for different polynomial orders. The relative error for $\gamma < 1$ is high, because the option price is close to zero. For small values of $M$, when the approximation is not optimal, the errors do not have to be strictly decreasing in $M$, which is expected.

For convenience, in Table~\ref{tab:BS_L2errors} we list also (arithmetic) averages of both $\AE$ (denoted $\AAE$) and $\RE$ (denoted $\ARE$) for two different sets of nodes:
\begin{enumerate}
\item $\gamma_i\in[0.7;1.3]$, $i=1,\dots,61$, taken with the equidistant step $\Delta\gamma = 0.01$,
\item $\gamma_i\in[1;1.5]$, $i=1,\dots,11$, taken with the equidistant step $\Delta\gamma = 0.05$.
\end{enumerate}
Convergence of $\AAE$ of both sets is depicted in Figure~\ref{fig:BS_convergence}.

\subsection{Heston model}\label{ssec:results:H}

We consider the following setting of Heston model. The parameters are chosen as in many examples in the book by \cite{Rouah13}
\begin{itemize}
\item initial variance $v_0 = 0.05$,
\item variance $v \in [0;0.5]$ discretized with the equidistant step $\Delta v = 0.005$,
\item mean reversion rate $\kappa = 5$,
\item long-run variance $\theta = 0.05$,
\item volatility of volatility $\tilde{\sigma} = 0.5$,
\item correlation $\rho = -0.8$,
\item the price of volatility risk $\lambda = 0$,
\item risk free interest rate $r= 0.03$,
\end{itemize}
and the parameters of the options are the same as for the BS model (Sec. \ref{options_param}). We also impose $x=\ln(S_t).$ Combinations of Hermite and Laguerre polynomials are chosen for the orthogonal polynomial expansion. Our numerical solution $u_{M,N}$ of the Heston PDE \eqref{H:PDE} is considered in the form \eqref{eq:solH}. Fourier coefficients for $\tau = T$ are obtained by solving the system of ODEs \eqref{H:system_odes} with $\tilde{B}$ given by \eqref{h:matrixB} and \eqref{H:Bab}. For consistency reasons, we use MATLAB ODE solver \texttt{ode45} to solve this system, that again leads to smaller values of $L^2$ error, although its speed is now much slower than the numerical calculation of matrix exponential in \eqref{H:expm} using MATLAB procedure \texttt{expm} \citep[Sec.~4.1]{Filipova19}.

In order to evaluate
\begin{align*}
u_{M,N} (x,v,T) &= \sum_{m=0}^{M} \sum_{n=0}^{N} c_{m,n}(T) H_m(x) L_n(v) 
= \sum_{n=0}^{N} L_n(v) \left( \sum_{m=0}^{M} c_{m,n}(T) H_m(x) \right),
\end{align*}
we repeatedly apply the Clenshaw's recurrence formula as indicated. 
To numerically evaluate the $L^2$ error, we make use of the Gauss-Hermite (with 251 Hermite points) and Gauss-Laguerre (with 201 Laguerre points) quadratures. For pointwise comparison we make use of the same set of nodes $(1)$ and $(2)$ as in Section~\ref{ssec:results:BS} with $v=v_0$.

In Figure \ref{fig:h_pde_ode45} on the left we can see the Heston-Lewis formula $u^{\text{H}}$ and numerical solution $u_{M,N}$ for $M = 35$ and $N = 30$ zoomed to the ITM region. This chosen combinations of polynomial orders present anticipated behaviour of the solution. The five dashed grid lines at the $xv$ plane are plotted at $\gamma\in\{1.1,1.2,1.3,1.4,1.5\}$. On the right we plot the relative error $|1-u_{M,N}/u^{\text{H}}|$. For the same polynomial orders we plot the absolute error $|u_{M,N}-u^{\text{H}}|$ and relative error $|1-|u_{M,N}/u^{\text{H}}$ for different values of $v$ in Figure \ref{fig:h_pde_ode45_aere}. 

\begin{figure}[ht]
\centering
\includegraphics[width=\textwidth,trim={15mm 0 25mm 0},clip]{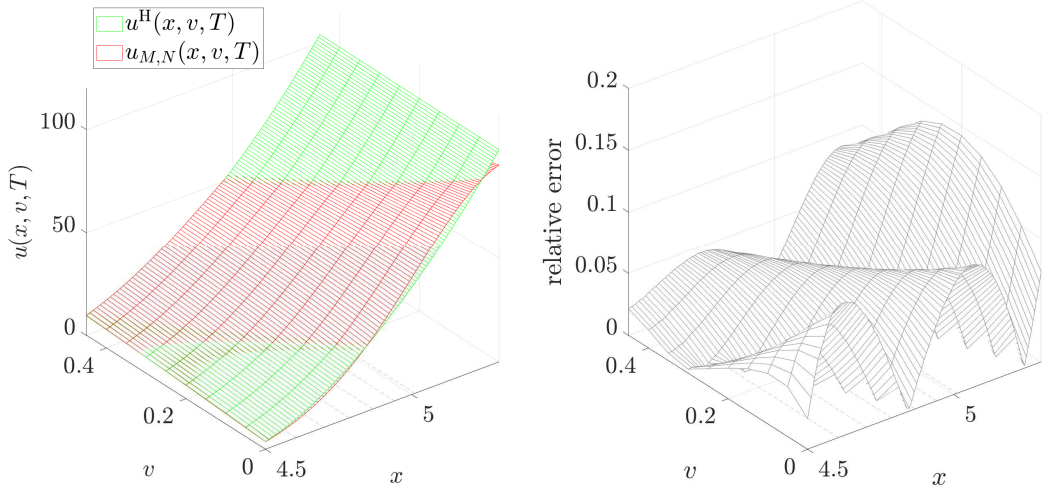}
\caption{Heston-Lewis formula $u^{\text{H}}(x,v,T)$ and the PDE solution $u_{M,N}(x,v,T)$ for $M=35$ and $N=30$ on the left, relative error $|1-u_{M,N}/u^{\text{H}}|$ on the right. The five dashed grid lines at the $xv$ plane are plotted at $\gamma\in\{1.1,1.2,1.3,1.4,1.5\}$. }
\label{fig:h_pde_ode45}
\end{figure} 

\begin{figure}[ht]
\centering
\includegraphics[width=\textwidth,trim={22mm 0 29mm 0},clip]{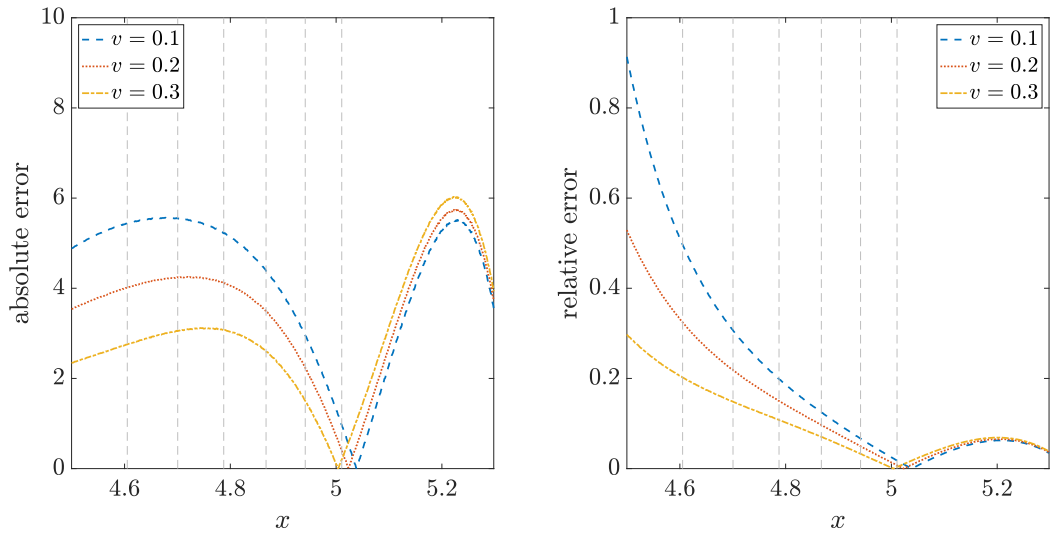}
\caption{Absolute and relative errors of the PDE solution $u_{M,N}(x,v,T)$ for $M=35$ and $N=30$ plotted for different values of $v$. The five dashed vertical grid lines are plotted at $\gamma\in\{1.1,1.2,1.3,1.4,1.5\}$.}
\label{fig:h_pde_ode45_aere}
\end{figure} 

Similarly as in the BS case, we measure the $L^2$ error, absolute error $\AE(\gamma)$ and relative error $\RE(\gamma)$ calculated at given point $x=\ln(\gamma K)$ and $v=v_0$, i.e. now 
\begin{align*}
\AE(\gamma) &= |u_{M,N}(\ln(\gamma K),v_0,T) - u^{\text{H}}(\ln(\gamma K),v_0,T)|, \\
\RE(\gamma) &= |1-u_{M,N}(\ln(\gamma K),v_0,T) / u^{\text{H}}(\ln(\gamma K),v_0,T)|,
\end{align*}
and also average errors $\AAE$ and $\ARE$ for the two set of nodes $(1)$ and $(2)$ as described above. All results are summarized in Tables \ref{tab:H_L2errors} and \ref{tab:H_AEREerrors}. Whereas the $L^2$ error is decreasing with increasing $M$, the influence of increasing $N$ is not significant that can be seen also in Figure \ref{fig:h_conv_pde} where we depicted the $\AAE$ for the ITM set of nodes $(2)$. From the numerical analysis point of view it is also interesting to mention that very few Laguerre points used in the numerical quadrature lie within the considered region $v\in[0;0.5]$ and experiments showed that the contribution from the majority of remaining points can be neglected.

\begin{figure}[ht]
\centering
\includegraphics[width=\textwidth,trim={22mm 0 29mm 0},clip]{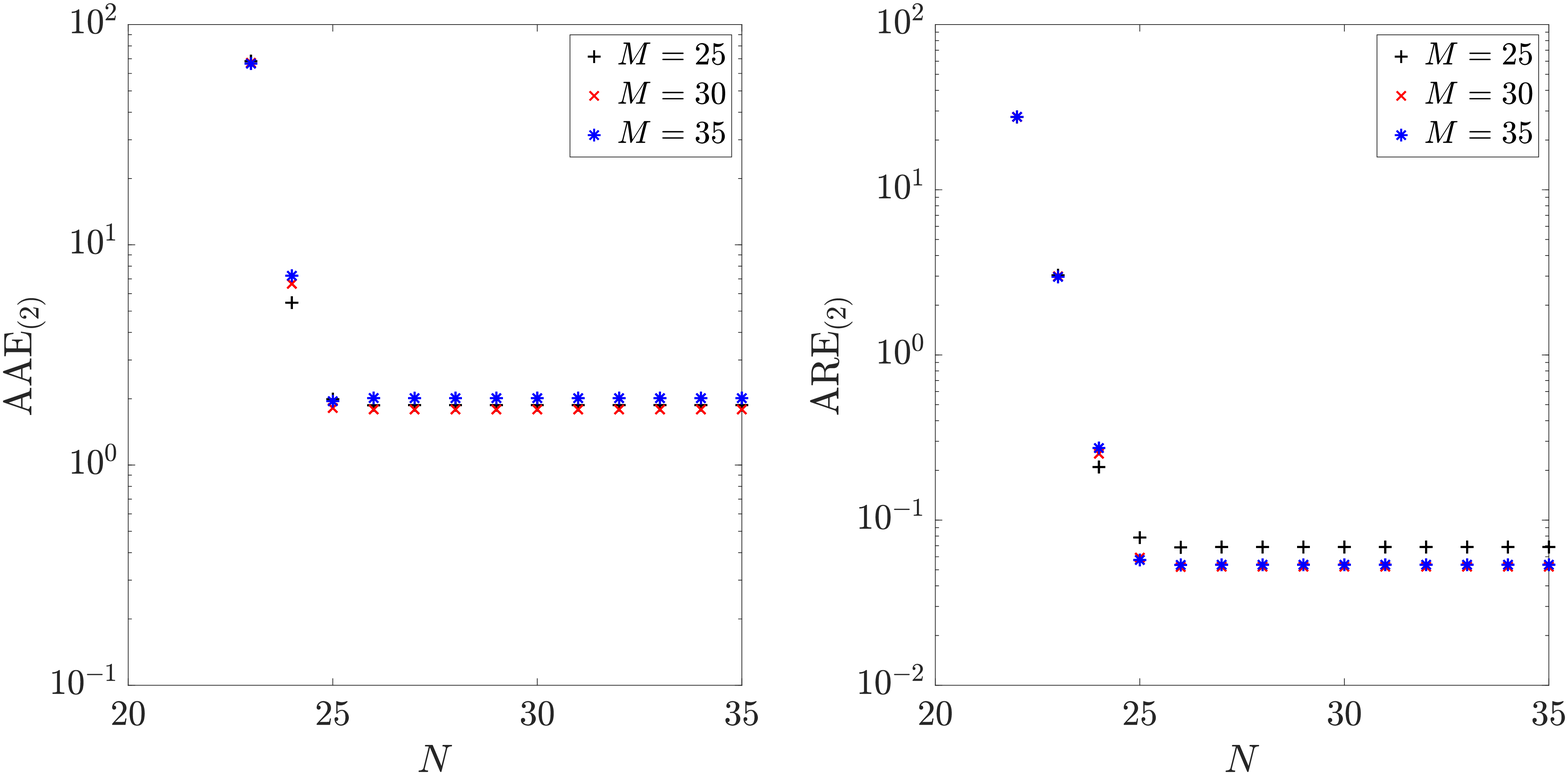}
\caption{Convergence of the average absolute error $\AAE$ and average relative error $\ARE$ for the set of nodes $(2)$.}
\label{fig:h_conv_pde}
\end{figure} 

\begin{table}[ht]
\centering
\caption{Errors comparison for Heston model for different polynomial orders: $L^2$ error, average absolute error AAE and average relative error ARE are listed for two sets of points 1 and 2.}\label{tab:H_L2errors}
\begin{tabular}{lllllll}
\toprule
$M$ & $N$ & $L^2$ error & $\AAE_{(1)}$ & $\ARE_{(1)}$ & $\AAE_{(2)}$ & $\ARE_{(2)}$ \\
\midrule
25 & 26 & 0.238643 & 1.09888 & 0.390292 & 1.86799 & 0.0683568 \\
25 & 28 & 1.1757e-06 & 1.1014 & 0.382943 & 1.87171 & 0.0686957 \\
25 & 30 & 8.39452e-07 & 1.1014 & 0.382948 & 1.8717 & 0.0686954 \\
30 & 26 & 0.238643 & 0.547496 & 0.288746 & 1.78832 & 0.0520266 \\
30 & 28 & 1.18026e-06 & 0.547673 & 0.281325 & 1.7877 & 0.0522312 \\
30 & 30 & 8.44246e-07 & 0.547673 & 0.28133 & 1.7877 & 0.052231 \\
35 & 26 & 0.238642 & 0.493863 & 0.244035 & 2.01489 & 0.0535306 \\
35 & 28 & 1.17799e-06 & 0.491256 & 0.236509 & 2.01206 & 0.0536491 \\
35 & 30 & 8.42292e-07 & 0.491258 & 0.236514 & 2.01206 & 0.0536491 \\
\bottomrule
\end{tabular}
\end{table}

\begin{table}[ht]
\centering
\caption{Errors comparison for Heston model for different polynomial orders: absolute error AE and relative error RE at several selected nodes are listed.}\label{tab:H_AEREerrors}
\begin{tabular}{llllllll}
\toprule
$M$ & $N$ & $\AE(0.7)$ & $\RE(0.7)$ & $\AE(1)$ & $\RE(1)$ & $\AE(1.3)$ & $\RE(1.3)$ \\
\midrule
25 & 26 & 0.511921 & 3.51694 & 0.92327 & 0.0909802 & 0.844315 & 0.024334 \\
25 & 28 & 0.497216 & 3.41591 & 0.936251 & 0.0922593 & 0.85614 & 0.0246748 \\
25 & 30 & 0.497225 & 3.41598 & 0.936242 & 0.0922585 & 0.856132 & 0.0246746 \\
30 & 26 & 0.386106 & 2.65258 & 0.539511 & 0.053164 & 0.939491 & 0.027077 \\
30 & 28 & 0.371401 & 2.55155 & 0.552491 & 0.0544432 & 0.927665 & 0.0267362 \\
30 & 30 & 0.37141 & 2.55162 & 0.552483 & 0.0544423 & 0.927673 & 0.0267364 \\
35 & 26 & 0.331769 & 2.27928 & 0.418323 & 0.0412221 & 1.90743 & 0.0549739 \\
35 & 28 & 0.317063 & 2.17825 & 0.431304 & 0.0425013 & 1.8956 & 0.0546331 \\
35 & 30 & 0.317073 & 2.17831 & 0.431296 & 0.0425004 & 1.89561 & 0.0546333 \\
\bottomrule
\end{tabular}
\end{table}

\FloatBarrier 

Following Section~\ref{sec:transport}, we now analyse the solution close to the boundary $v=0$. We consider $h=0.005$ and polynomial orders $M = 35$ and $N = 30$. In Figure~\ref{fig:h_near_v_0} on the left we can see the Heston--Lewis formula $u^{\text{H}}$ for $v = 0$, PDE solution $u_{M,N}$ at $v=0$ and the boundary solution $u^{\text{B}}$ gained by application of theory in Section~\ref{sec:transport} for $h=0.005$. The vertical dashed grid lines are plotted again at $\gamma\in\{0.7,1,1.3\}$. On the right we can see how the transport equation solution differs to the two remaining ones.

\begin{figure}[ht]
\centering
\includegraphics[width=\textwidth,trim={22mm 0 29mm 0},clip]{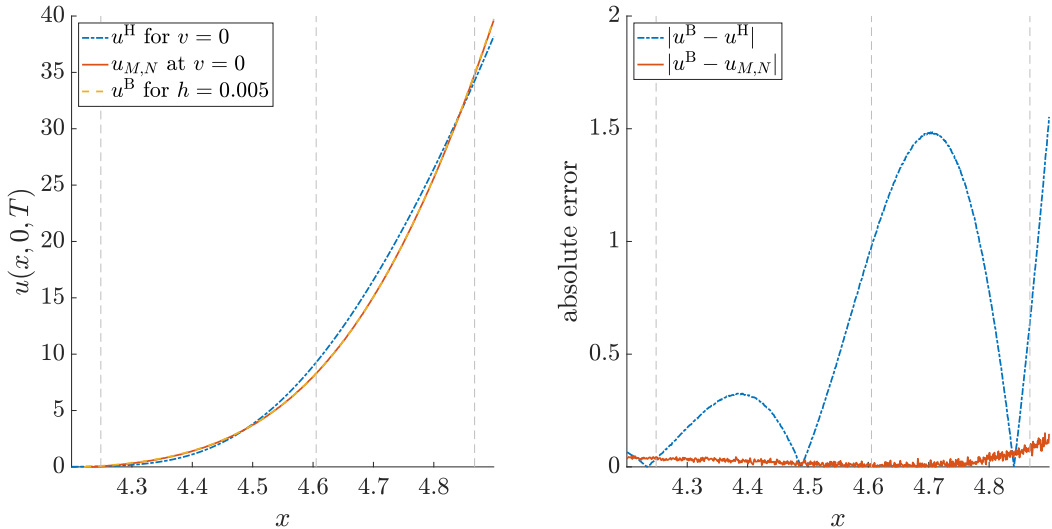}%
\caption{Comparison of the solution behaviour near the boundary $v = 0$ with zoom to the ATM region. The vertical grid lines are plotted at $\gamma\in\{0.7,1,1.3\}$.}
\label{fig:h_near_v_0}
\end{figure}

\section{Conclusion}\label{sec:conclusion}

The analyticity of the solution of the Heston model has been shown in the recent paper of \cite{Alziary18}. A crucial step in their proof is the approximation of the payoff by a sequence of entire functions, in particular Hermite and Laguerre functions \citep[sec. 11.1]{Alziary18}, with the Galerkin method \citep[sec. 11.2]{Alziary18}. The aim of our paper was to make use of these theoretical results to study an alternative method for the option pricing problem for the \cite{BlackScholes73} model and the \cite{Heston93} model. Moreover, we were interested in the behaviour of the solution near the zero volatility boundary and considered the equation for vanishing volatility. This approach was also motivated and theoretically justified by results of \cite{Alziary20pre}.

By the numerical implementation of Galerkin's method in weighted Sobolev spaces we found an alternative representation of the solution to both BS and Heston models. The obtained representation is a smooth approximation of the solution that does not share the serious numerical difficulties of existing semi-closed formulas as they were presented by \cite{DanekPospisil20ijcm}. Presented approach is also independent of the space variable discretization (used for example by Galerkin finite element methods) or spacial node locations (needed by radial basis function methods).

The presented experiments give a first insight into the performance of the method but thorough numerical analysis has to be performed in order to properly understand the behaviour of the solutions for higher polynomial orders. There are different possibilities how one could try to improve the method, for example to use other procedures to solve the system of ordinary differential equations especially such that take into consideration the specific triangular form of the matrix. A detailed error analysis and the application of additional procedures were beyond the scope of this paper and is left as an open issue.

A considerable advantage of the presented approach is that it can be easily adapted to other stochastic volatility models by following the steps at the beginning of Section~\ref{sec:HestonPDE} and using the calculations of Theorem~\ref{t:int}. Aside from that, different payoff functions can be used as long as they are in the weighted Lebesgue space which applies to most of the generally used payoffs.

\section*{Funding}

The work was partially supported by the Czech Science Foundation (GA\v{C}R) grant no. GA18-16680S ``Rough models of fractional stochastic volatility''.

\section*{Acknowledgements}

This work is based on the Master's thesis \cite{Filipova19} titled \emph{Solution of option pricing equations using orthogonal polynomial expansion} that was written by Kate\v{r}ina Filipov\'{a} and supervised by Jan Posp\'{\i}\v{s}il. The thesis was also advised by Falko Baustian during the two months internship of Kate\v{r}ina Filipov\'{a} at the University of Rostock.

Our sincere gratitude goes to Prof. Peter Tak\'{a}\v{c} from the University of Rostock, who introduced us to the problem and provided us with valuable suggestions and insightful criticism, and to both anonymous referees for their valuable comments and extensive suggestions.

Computational resources were provided by the CESNET LM2015042 and the CERIT Scientific Cloud LM2015085, provided under the programme ``Projects of Large Research, Development, and Innovations Infrastructures''.




\end{document}